\documentclass[pdflatex,sn-mathphys-num,iicol]{sn-jnl} 
\usepackage{amsmath} 
\usepackage{amssymb}
\usepackage{graphicx} 
\usepackage{xcolor} 
\usepackage{footmisc}
\usepackage[ nameinlink, capitalise, noabbrev ]{cleveref} 
\hypersetup{ colorlinks=true, linkcolor=blue, citecolor=blue, urlcolor=blue } 
\usepackage{placeins} % to \clearpage without starting a new page - with \FloatBarrier

\usepackage{ulem} %needed for deleted and added commands 
\usepackage{comment}
\newcommand{\vs}[0]{\mkern1mu} % very small space in math mode - vector space between vector and ) or |

\usepackage{svg}
\usepackage{tikz}
\usetikzlibrary{angles,quotes,calc,shapes.geometric, arrows.meta, positioning }

\newtheorem{proposition}{Proposition}

\newcommand{\Es}{E^{*}}          % feature-enhanced field
\newcommand{\xv}{\vec{x}}
\newcommand{\yv}{\vec{y}}

\begin{document} 
\title[Why the Multi-Sphere Shape Generator Works]{Why the Multi-Sphere Shape Generator Works}
\subtitle{Medial-Axis Placement of Spheres} 
\author*[1]{\fnm{Arash} \sur{Moradian}
%\email{arash.moradian@fau.de}
} 
\author[1]{\fnm{Felix} \sur{Buchele}} 
\author[1]{\fnm{Thorsten} \sur{P\"oschel}
%\orcid{0000-0003-1285-3863} 
}
\affil[1]{\orgdiv{Lehrstuhl f\"ur Multiskalensimulation}, 
\orgname{Friedrich-Alexander-Universit\"at Erlangen-N\"urnberg}, 
\country{Germany}} 
\abstract{The Multi-Sphere Shape Generator (MSS) \cite{buchele2026mss} places spheres according to a feature-enhanced residual field, but the geometric basis of this strategy has remained unknown. We prove that every local maximum of the residual field lies on the medial axis of the target shape, implying that MSS places spheres at the centers of maximal inscribed spheres without explicit skeleton extraction. Numerical tests show that deviations from the exact medial axis are limited to the voxel resolution. This result provides a mathematical explanation for the placement strategy that underlies the accuracy of MSS.}

\maketitle 

\section{Introduction} 
\label{sec:intro} 

%\begin{itemize}
%    \item Why Multi-Sphere Approximation?
%    \item Why is good sphere placement important?
%    \item ??? Previous methods
%    \begin{itemize}
%        \item clustering
%        \item medial-axis extraction
%        \item optimization
%    \end{itemize}
%    \item MSS
%    \item \textbf{Open question: Why does MSS place spheres so well?}
%\end{itemize}

Multi-sphere particle models were introduced \cite{poschel1993static, Buchholtz.1994} as an efficient way to represent non-spherical particles while retaining the simplicity of sphere-sphere contact detection \cite{Kruggel:2008, Markauskas:2010}. Since then, multi-sphere models have become a standard approach in DEM, and numerous methods have been proposed for their automatic generation \cite{Favier:1999, Ferellec:2008, Lu:2015, angelidakis2021clump}. Their accuracy depends not only on the number of spheres used, but also on how effectively these spheres are positioned within the particle. Consequently, the development of algorithms for automatic sphere placement has received considerable attention.

Existing generators follow different strategies, including clustering-based methods
\cite{Garcia.2009},
medial-axis-based methods
\cite{Yuan:2019},
and optimization approaches
\cite{Li:2015}.
Medial-axis methods are particularly attractive because the centers of maximal inscribed spheres lie on the medial axis of the object \cite{blum1967transformation, siddiqi2008medial}, providing a compact geometric description of the particle shape. Accordingly, medial-axis and skeleton extraction have been used for the generation of multi-sphere models \cite{Yuan:2019}. However, explicit computation of the medial axis can be computationally demanding and is sensitive to discretization and boundary perturbations \cite{Attali:2009, siddiqi2008medial}.

The Multi-Sphere Shape Generator (MSS), introduced in \cite{buchele2026mss}, is a greedy algorithm that places spheres according to a feature-enhanced residual field. Without explicitly extracting the medial axis, MSS produces accurate multi-sphere approximations using comparatively few spheres. While its performance has been demonstrated numerically \cite{buchele2026mss, gemss2026, Buchele.2026}, the geometric basis of its sphere-placement strategy has not yet been established.

In this paper, we provide the missing theoretical explanation. We prove that every local maximum of the residual field used by MSS lies on the medial axis of the target shape in the continuum limit. Consequently, MSS places sphere centers at the centers of maximal inscribed spheres without explicit medial-axis extraction. Numerical experiments confirm that deviations from the exact medial axis are solely due to voxel discretization. 
%\deleted{These results provide the first mathematical explanation for the sphere-placement strategy underlying MSS.}

\section{Multi-Sphere Shape Generator - MSS}
\label{sec:mss}

%explain MSS
%\subsection{Distance-transform-based sphere placement}
%\begin{itemize}
%   \item EDT
%   \item sphere radius
%   \item greedy
%\end{itemize}

%\subsection{Residual field of MSS}
%\begin{equation}
%    \hat{E}^\ast=2\hat{E}-\hat{\bar{E}}
%\end{equation}

%\cref{fig:fields}: Comparison of driving fields used for greedy sphere placement.
%\begin{itemize}
%    \item (a) Euclidean distance field.
%    \item (b) Conventional residual field, whose maximum shifts toward the remaining uncovered boundary.
%    \item (c) MSS feature-enhanced residual field, whose maximum is located near the center of the remaining %feature.
%\end{itemize}

The multi-sphere approach represents a target shape $\mathcal{S}$ by the union
of $n$ spheres,
\begin{equation}
    \tilde{\mathcal{S}}
    =
    \bigcup_{i=0}^{n-1}
    \tilde{\mathcal{S}}_i\left(\xv_i,R_i\right),
    \label{eq:multisphere}
\end{equation}
where $\xv_i$ and $R_i$ denote the center and radius of sphere $i$,
respectively.

\subsection{Inscribed spheres and the Euclidean Distance Transform}
\label{sec:distance}

For $\xv\in\mathcal{S}$, the Euclidean Distance Transform
\begin{equation}
    E\left(\xv\right)
    =
    \min_{\yv\in\partial\mathcal{S}}
    \left|\xv-\yv\vs\right|
    \label{eq:edt}
\end{equation}
gives the distance to the boundary $\partial\mathcal{S}$. Hence, $E\left(\xv\right)$ is
the radius of the largest sphere centered at $\xv$ that is contained in
$\mathcal{S}$. Once a sphere center $\xv_i$ has been selected, its radius is
therefore
\begin{equation}
    R_i = E\left(\xv_i\right)\,.
    \label{eq:radius}
\end{equation}
Constructing the multi-sphere model thus reduces to selecting the sphere
centers $\xv_i$.

\subsection{Residual field and sphere placement}
\label{sec:residual}

MSS selects these centers iteratively. Let $\tilde E\left(\xv\,\right)$ denote the Euclidean Distance Transform of the current multi-sphere approximation $\tilde{\mathcal{S}}$. MSS defines the residual field
\begin{equation}
    E^\ast\left(\xv\vs\right)
    =
    2E\left(\xv\vs\right)-\tilde E\left(\xv\vs\right)
    \label{eq:mss_residual}
\end{equation}
and places the next sphere at a maximum of $E^\ast$. Its radius is then given by \cref{eq:radius}. The newly placed sphere is added to $\tilde{\mathcal S}$, the residual field is updated, and the procedure is repeated until a prescribed stopping criterion is reached. For a detailed description of the MSS algorithm, we refer to \cite{buchele2026mss}.

The effect of this residual field is illustrated in \cref{fig:fields} for a simple two-disk example. After the first disk has been represented, the maximum of a conventional masked-residual distance field \cite{angelidakis2021clump} shifts away from the center of the remaining disk, whereas the maximum of the MSS residual field remains at its center. The central question of this work is whether this behavior is specific to this simple example or follows from the structure of the MSS residual field in general.

\begin{figure}[htb]
  \centering
  \includegraphics[width=0.9\columnwidth]{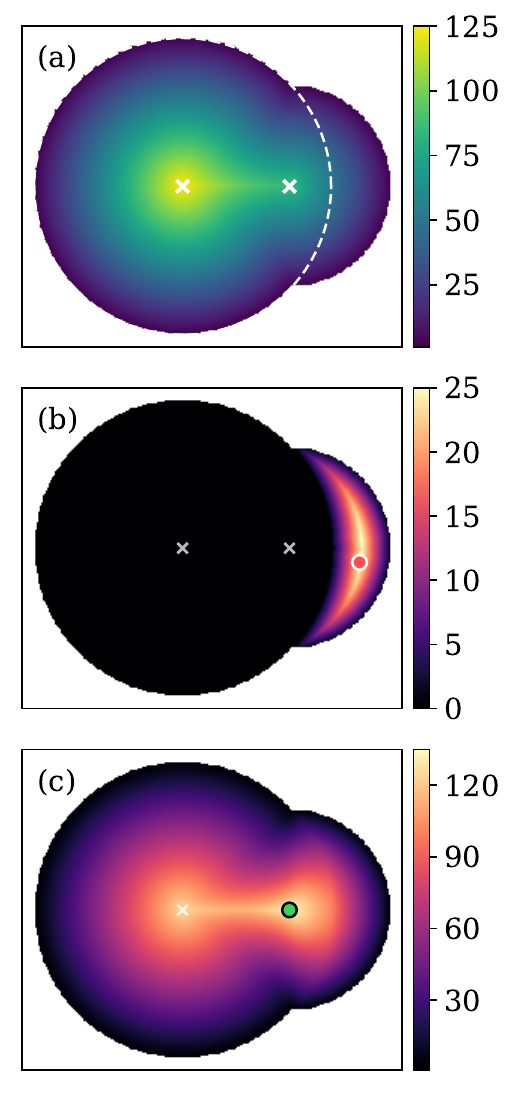}
  \caption{Comparison of driving fields for greedy sphere placement in the
  two-disk example.
  (a) Euclidean Distance Field $E$.
  (b) Conventional masked-residual field, whose maximum shifts away from the
  center of the remaining feature.
  (c) MSS residual field $2E-\tilde E$, whose maximum remains at the center of
  the remaining feature.}
  \label{fig:fields}
\end{figure}

%\section{Why Medial-Axis Placement Matters}
%\textbf{Why is the medial axis of any interest?}
%\begin{itemize}
%    \item maximum inscribed spheres
%    \item centers of maximal inscribed spheres lie on the medial axis
%    \item off-axis spheres are never maximal
%    \item off-axis spheres can always be replaced by a larger sphere
%    \item efficient multi-sphere approximations therefore require sphere centers close to the medial axis
%\end{itemize}
%\textbf{Open question: Does MSS place sphere centers on the medial axis?} = End of section

\section{Why medial-axis placement matters}
\label{sec:medial-axes}

%MSS does not explicitly construct the medial axis. Nevertheless, the two-disk example in \cref{fig:fields} suggests that the maxima of its residual field may lie on the medial axis. In the following, we show that this is indeed a general property of MSS in the continuum limit.

The medial axis of a shape $\mathcal S$ is the locus of the centers of its maximal inscribed spheres \cite{blum1967transformation, siddiqi2008medial}. An inscribed sphere is maximal if it is not strictly contained in any other inscribed sphere. Hence, every inscribed sphere whose center lies outside the medial axis is contained in a larger inscribed sphere centered on the medial axis.

This property is particularly relevant to multi-sphere approximation. An off-axis sphere is geometrically redundant in the sense that the region it covers is also covered by a larger inscribed sphere centered on the medial axis. Medial-axis placement therefore provides a geometrically efficient choice of sphere centers without restricting the region that can be represented.

The consequence for sphere placement is illustrated by the pill-shaped target in \cref{fig:pill}. With conventional residual-based placement, the remaining uncovered regions lead to two additional off-axis spheres. In contrast, MSS places a single sphere on the medial axis, where the largest inscribed sphere simultaneously covers both regions.
\begin{figure}[htb]
    \centering
    \includegraphics[width=0.9\columnwidth]{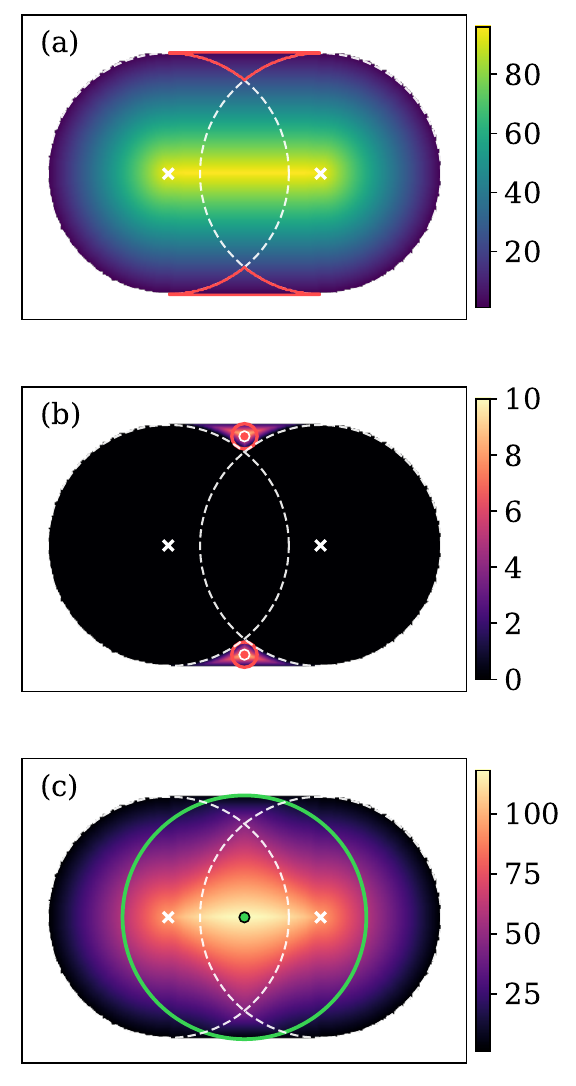}
    \caption{Residual fields on a pill-shaped target. (a)~$E$.
  (b)~Distance transform of the masked residual $S\setminus\tilde S$: separate maxima at each gap require two additional spheres. (c)~Feature-enhanced field $\Es$: maximum (green) falls on the medial axis, covering both gaps with a single sphere.}
  
    \label{fig:pill}
\end{figure}

MSS does not explicitly construct the medial axis. Nevertheless, the two-disk example in \cref{fig:fields} and the pill-shaped example in \cref{fig:pill} suggest that the maxima of its residual field may lie on the medial axis. In the following, we show that this is indeed a general property of MSS in the continuum limit.

%\section{Medial-Axis Placement of MSS}
%\subsection{Preliminaries}
%\begin{itemize}
%    \item Euclidean Distance Field
%    \item Lipschitz continuit
%    \item gneralized Clarke gradient
%\end{itemize}

%\subsection{??? theorem}
%\textbf{Every local maximum of the MSS residual field lies on the medial axis.}
%proof

%\subsection{Implications}
%\begin{itemize}
%    \item MSS does not explicitly compute the medial axis.
%    \item Nevertheless, every newly inserted sphere is centered on the medial axis.
%    \item The theorem explains the excellent sphere placement observed in practice.
%\end{itemize}

\section{Medial-axis placement of MSS}
\label{sec:theory}

We now show that the sphere-placement rule of MSS necessarily selects medial-axis locations in the continuum limit. The argument relies on two properties of Euclidean Distance Fields. Away from the medial axis, the distance field
$E$ is differentiable and satisfies the eikonal equation
\begin{equation}
    \left|\nabla E\left(\xv\vs\right)\right|=1\,.
    \label{eq:eikonal}
\end{equation}
Moreover, the distance field $\tilde E$ of the current multi-sphere approximation is $1$-Lipschitz,
\begin{equation}
    \left|\tilde E\left(\xv\vs\right)
    -\tilde E\left(\yv\vs\right)\right|
    \leq
    \left|\xv-\yv\vs\right|,
    \label{eq:lipschitz}
\end{equation}
and therefore, wherever it is differentiable,
\begin{equation}
    \left|\nabla\tilde E\left(\xv\vs\right)\right|\leq 1\,.
    \label{eq:gradient_bound}
\end{equation}

These two properties already reveal the mechanism underlying MSS. At any point away from the medial axis where both distance fields are differentiable, the gradient of the residual field
\begin{equation}
    E^\ast\left(\xv\vs\right)
    =
    2E\left(\xv\vs\right)-\tilde E\left(\xv\vs\right)
\end{equation}
is
\begin{equation}
    \nabla E^\ast
    =
    2\nabla E-\nabla\tilde E\,.
    \label{eq:residual_gradient}
\end{equation}
The reverse triangle inequality together with \cref{eq:eikonal,eq:gradient_bound} gives
\begin{equation}
    \left|\nabla E^\ast\right|
    \geq
    2\left|\nabla E\right|
    -\left|\nabla\tilde E\right|
    \geq 1\,.
    \label{eq:residual_gradient_bound}
\end{equation}
Hence, $E^\ast$ has no stationary point away from the medial axis wherever both distance fields are differentiable.

To include non-differentiable points, we use the Clarke generalized gradient $\partial_C f$ \cite{clarke1990optimization}, defined as the convex hull of the limiting gradients at nearby differentiable points,
\begin{equation}
\begin{aligned}
    \partial_C f\left(\xv\vs\right)
    &=
    \operatorname{co}
    \left\{
        \lim_{k\to\infty}\nabla f\left(\xv_k\right)
        :
        \xv_k\to\xv,\;
    \right.
    \\
    &\qquad\qquad
        \left.
        f\text{ differentiable at }\xv_k
        \vphantom{\lim_{k\to\infty}\nabla f\left(\xv_k\right)}
    \right\},
\end{aligned}
\label{eq:clarke}
\end{equation}
where $\operatorname{co}$ denotes the convex hull.

\begin{proposition}
\label{prop:medial}
Every local maximum of the MSS residual field
$E^\ast=2E-\tilde E$ in the interior of $\mathcal S$ lies on the medial axis
of $\mathcal S$.
\end{proposition}

\begin{proof}
Suppose that $\xv^\ast$ is a local maximum of $E^\ast$ that does not lie on
the medial axis. Then $E$ is differentiable at $\xv^\ast$ and
\begin{equation}
    \left|\nabla E\left(\xv^\ast\right)\right|=1\,.
    \label{eq:proof_eikonal}
\end{equation}
Since $\tilde E$ is a distance function, it is $1$-Lipschitz, every element of its Clarke generalized gradient, $\vec{p}\in\partial_C\tilde E\left(\xv^\ast\right)$, satisfies
\begin{equation}
    \left|\vec{p}\vs\vs\right|\leq 1\,.
    \label{eq:proof_p_bound}
\end{equation}
The Clarke generalized gradient of the residual field satisfies
\begin{equation}
    \partial_C E^\ast\left(\xv^\ast\right)
    \subseteq
    \left\{
        2\nabla E\left(\xv^\ast\right)-\vec{p}
        :
        \vec{p}\in
        \partial_C\tilde E\left(\xv^\ast\right)
    \right\}.
    \label{eq:proof_clarke_residual}
\end{equation}
For every such $\vec{p}$,
\begin{equation}
    \left|
        2\nabla E\left(\xv^\ast\right)-\vec{p}\vs
    \right|
    \geq
    2\left|\nabla E\left(\xv^\ast\right)\right|
    -\left|\vec{p}\vs\right|
    \geq 1\,.
    \label{eq:proof_bound}
\end{equation}
Therefore,
\begin{equation}
    0\notin\partial_C E^\ast\left(\xv^\ast\right).
\end{equation}
This contradicts the Clarke necessary condition
\begin{equation}
    0\in\partial_C E^\ast\left(\xv^\ast\right)
\end{equation}
for a local maximum of a locally Lipschitz function \cite{clarke1990optimization}. Hence,
$\xv^\ast$ must lie on the medial axis.
\end{proof}

\section{Numerical verification}
\label{sec:verification}
%\subsection{box}
%\subsection{Accuracy of medial-axis placement}
%?????

The result of \cref{prop:medial} applies to continuous distance fields, whereas MSS operates on a voxelized representation of the target shape. Discretization may therefore cause the computed sphere centers to deviate from the exact medial axis. We quantify this deviation for a shape whose medial axis is known analytically and examine its dependence on the voxel size.

The test shape is an axis-aligned box 
\begin{equation}
    \mathcal{S}
    =
    [-a,a]\times[-b,b]\times[-c,c]
\end{equation}
with side lengths $2a=30$, $2b=20$, and $2c=10$. For a box, the medial axis consists of the points at which at least two nearest faces are equidistant. It can therefore be determined analytically, independently of the voxelization or any numerical medial-axis extraction computation.

For each sphere center $\xv_i$ generated by MSS, we determine its shortest distance $d_i$ from the analytical medial axis $\mathcal{M}$. 
\begin{equation}
    d_i
    =
    \min_{\yv\in\mathcal{M}}
    \left|\xv_i-\yv\vs\right|\,.
    \label{eq:medial_distance}
\end{equation}
\begin{figure}[htb]
    \centering
    \includegraphics[width=0.9\columnwidth]{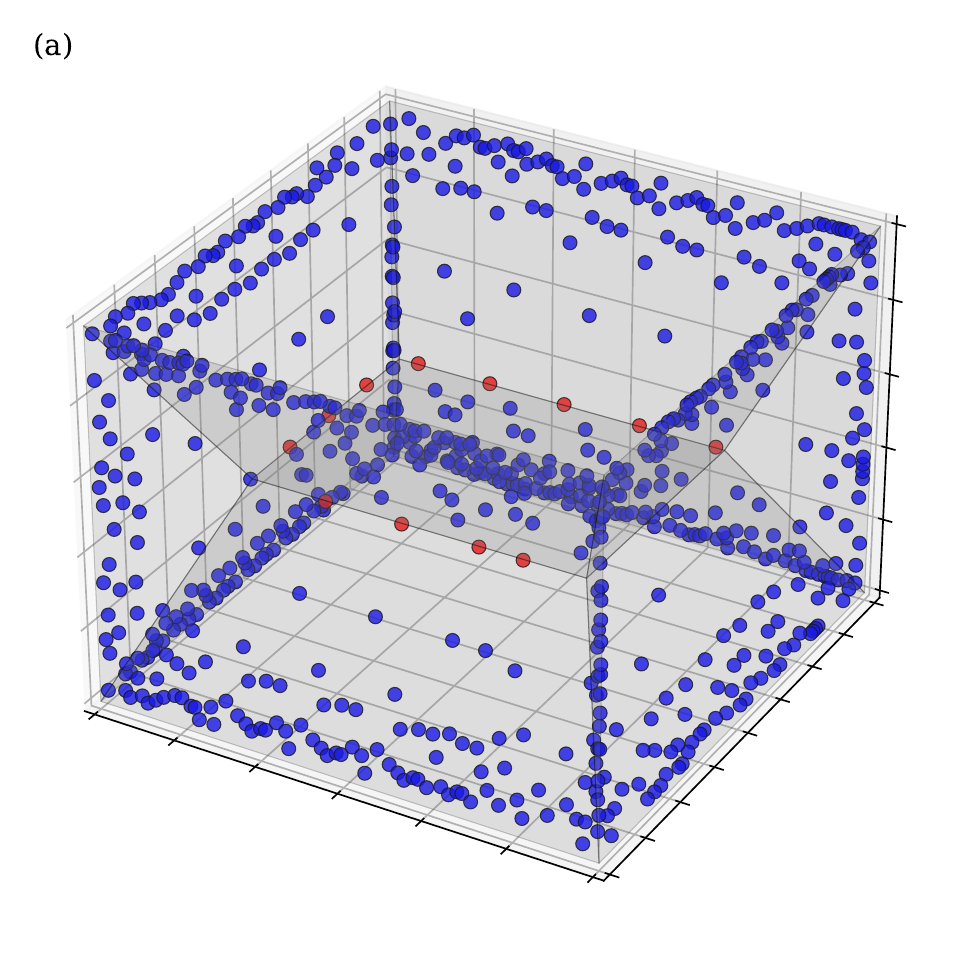}
    \includegraphics[width=0.9\columnwidth]{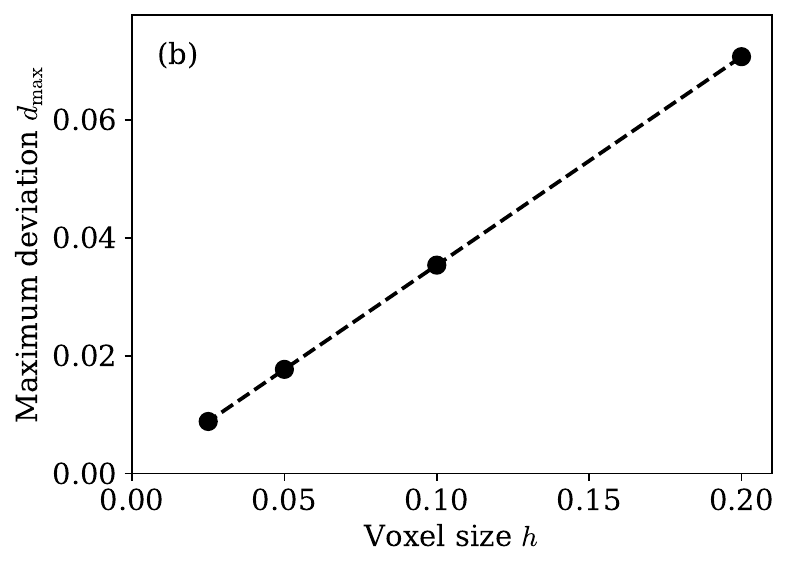}
    \caption{Numerical verification of medial-axis placement for an axis-aligned box. (a) MSS sphere centers relative to the analytical medial axis at voxel size $h=0.1$. The blue symbols mark sphere centers that are on the medial axis, within numerical tolerance. The red symbols mark particle centers that are $0.345h$ away. (b) Maximum deviation $d_{\max}$ from the analytical medial axis as a function of the voxel size $h$. The dashed line indicates linear scaling, $d_{\max}\propto h$.
    }
    \label{fig:box}
\end{figure}

\Cref{fig:box}a shows the sphere centers relative to the analytical
medial axis for a representative voxel size $h=0.1$. At this resolution, $97.2\%$ of the sphere centers coincide with the analytical medial axis within a numerical tolerance of $10^{-9}$, while the maximum deviation is $0.354\,h$.

To examine the effect of spatial discretization, the calculation is repeated for voxel sizes $h\in \{0.2, 0.1, 0.05, 0.025\}$, while keeping the dimensions of the box and the remaining MSS parameters unchanged. For each resolution, we determine the maximum deviation
\begin{equation}
    d_{\max}
    =
    \max_i d_i\,.
    \label{eq:dmax}
\end{equation}
As shown in \cref{fig:box}b, $d_{\max}$ decreases approximately linearly with the voxel size, 
\begin{equation}
    d_{\max}\propto h\,.
\end{equation}
Thus, the maximum deviation tends to zero under grid refinement, consistent with the continuum result of \cref{prop:medial}.

\section{Conclusion}
\label{sec:conclusion}

The result of \cref{prop:medial} explains a central feature of the sphere-placement strategy used by MSS. Although the algorithm does not construct the medial axis explicitly, the maxima of its residual field are restricted to the medial axis in the continuum limit. MSS therefore obtains the principal advantage of medial axis-based sphere placement without requiring a separate medial-axis extraction step.

The origin of this property lies in the particular residual field used by MSS. Away from the medial axis, the distance field of the target shape satisfies $\left|\nabla E\right|=1$, whereas the gradient magnitude of the distance field of the current multi-sphere approximation cannot exceed unity. The factor of two in $E^\ast=2E-\tilde E$ therefore prevents the residual field from having a local maximum away from the medial axis. 
This provides a mathematical explanation for the sphere-placement rule introduced with MSS.

The numerical results show how this continuum property is approached in the voxelized implementation. The deviations from the analytical medial axis decrease with the voxel size and tend to zero under grid refinement. The observed off-axis displacements are therefore consistent with spatial discretization rather than an intrinsic limitation of the MSS placement rule.

The present analysis is deliberately restricted to the residual field used by MSS. The underlying argument can be extended to a broader class of residual fields and, more generally, to functions satisfying the eikonal equation. These generalizations are considered separately in \cite{Arash:2026}.

\section*{Statements and Declarations} 

\bmhead{Funding} This research was funded by the Deutsche Forschungsgemeinschaft (DFG, German Research Foundation) – 377472739/GRK 2423/1-2019. The authors are very grateful for this support.

%\bmhead{Acknowledgements} \TP{whom to acknowledge?}

\bmhead{Competing Interests} The authors declare that they have no competing interests. 

\bibliography{MSS-math} 

@article{Ferellec:2008,
author = {Ferellec, J.-F. and McDowell, G. R.},
title = {A simple method to create complex particle shapes for {DEM}},
journal = {Geomechanics and Geoengineering},
volume = {3},
pages = {211--216},
year = {2008},
doi = {10.1080/17486020802253992},
}

@article{Attali:2009,
  author    = {Attali, Dominique and Boissonnat, Jean-Daniel and Edelsbrunner, Herbert},
  title     = {Stability and Computation of Medial Axes---{A} State-of-the-Art Report},
  journal   = {Mathematical Foundations of Scientific Visualization, Computer Graphics, and Massive Data Exploration},
  year      = {2009},
  pages     = {109--125},
  doi       = {10.1007/978-3-642-01538-$5_6$}
}

@article{Li:2015,
author = {Li, Cheng-Qing and Xu, Wen-Jie and Meng, Qing-Shan},
title = {Multi-sphere approximation of real particles for {DEM} simulation based on a modified greedy heuristic algorithm},
journal = {Powder Technology},
volume = {286},
pages = {478-487},
year = {2015},
doi = {10.1016/j.powtec.2015.08.026}
}

@article{Garcia.2009,
  author = {Garcia, X. and Latham, J.-P. and Xiang, J. and Harrison, J. P.},
  title = {A clustered overlapping sphere algorithm to represent real particles in discrete element modelling},
  journal = {G\'eotechnique},
  volume = {59},
  pages = {779--784},
  year = {2009},
  doi = {10.1680/geot.8.T.037}
}

@article{Yuan:2019,
author = {Yuan, Fei-Liang},
doi = {10.1007/s10035-019-0874-x}, 
year = {2019}, 
volume = {21}, 
title = {Combined 3D thinning and greedy algorithm to approximate realistic particles with corrected mechanical properties}, 
journal = {Granular Matter} 
}

@article{Favier:1999,
author = {Favier, J. F. and Abbaspour‐Fard, M. H. and Kremmer, M. and Raji, A. O.},
title = {Shape representation of axi‐symmetrical, non‐spherical particles in discrete element simulation using multi‐element model particles},
journal = {Engineering Computations},
volume = {16},
pages = {467-480},
year = {1999},
doi = {10.1108/02644409910271894}
}

@article{Markauskas:2010,
author = {Markauskas, D. and Ka{\v{c}}ianauskas, R. and D{\v{z}}iugys, A. and Navakas, R.},
title = {Investigation of adequacy of multi-sphere approximation of elliptical particles for {DEM} simulations},
journal = {Granular Matter},
year = {2010},
volume = {12},
page = {107-123},
doi= {10.1007/s10035-009-0158-y}
}

@article{Lu:2015,
author = {Lu, G. and Third, J. R. and M\"uller, C. R.},
title = {Discrete element models for non-spherical particle systems: {F}rom theoretical developments to applications},
journal = {Chemical Engineering Science},
volume = {127},
pages = {425-465},
year = {2015},
doi = {10.1016/j.ces.2014.11.050}
}

@article{Kruggel:2008,
title = {A study on the validity of the multi-sphere {D}iscrete {E}lement {M}ethod},
author = {Kruggel-Emden, H. and Rickelt, S. and Wirtz, S. and Scherer, V.},
journal = {Powder Technology},
volume = {188},
pages = {153-165},
year = {2008},
doi = {10.1016/j.powtec.2008.04.037}
}

@article{buchele2026mss,
  author  = {Buchele, Felix and P{\"o}schel, Thorsten and M{\"u}ller, Patric},
  title   = {Multi-sphere shape generator for {DEM} simulations of complex-shaped particles},
  journal = {Powder Technology},
  volume  = {480},
  pages   = {122635},
  year    = {2026},
  doi     = {10.1016/j.powtec.2026.122635}
}

@article{poschel1993static,
  author  = {P{\"o}schel, Thorsten and Buchholtz, Volkhard},
  title   = {Static friction phenomena in granular materials: {C}oulomb law versus particle geometry},
  journal = {Physical Review Letters},
  volume  = {71},
  xnumber  = {24},
  pages   = {3963--3966},
  year    = {1993},
  doi     = {10.1103/PhysRevLett.71.3963}
}

@article{Buchholtz.1994,
author = {Buchholtz, Volkhard and Pöschel, Thorsten},
 doi = {10.1016/0378-4371(94)90467-7},
 journal = {Physica A-Statistical Mechanics and Its Applications},
 pages = {390-401},
 title = {Numerical investigations of the evolution of sandpiles},
 volume = {202},
 year = {1994},
}

@article{angelidakis2021clump,
  author  = {Angelidakis, Vasileios and Nadimi, Sadegh and Otsubo, Masahide and Utili, Stefano},
  title   = {{CLUMP}: A code library to generate universal multi-sphere particles},
  journal = {SoftwareX},
  volume  = {15},
  pages   = {100735},
  year    = {2021},
  doi     = {10.1016/j.softx.2021.100735}
}

@incollection{blum1967transformation,
  author    = {Blum, Harry},
  title     = {A transformation for extracting new descriptors of shape},
  booktitle = {Models for the Perception of Speech and Visual Form},
  editor    = {Wathen-Dunn, Weiant},
  pages     = {362--380},
  address = {Cambridge, USA},
  publisher = {MIT Press},
  pages = {362-380},
  year      = {1967}
}

@book{siddiqi2008medial,
  editor    = {Siddiqi, Kaleem and Pizer, Stephen M.},
  title     = {Medial Representations: Mathematics, Algorithms and Applications},
  publisher = {Springer},
  address = {Dortrecht},
  series    = {Computational Imaging and Vision},
  volume    = {37},
  doi       = {10.1007/978-1-4,020-8658-8},
  year      = {2008}
}

@book{clarke1990optimization,
  author    = {Clarke, Frank H.},
  title     = {Optimization and Nonsmooth Analysis},
  publisher = {Society for Industrial and Applied Mathematics (SIAM)},
  address   = {Philadelphia},
  year      = {1990},
  series    = {Classics in Applied Mathematics},
  volume    = {5},
  doi       = {10.1137/1.9781611971309}
}

@article{Arash:2026,
    author = {Moradian, Arash and P{\"o}schel, Thorsten},
    title = {Residual Fields of Eikonal Functions Recover the Medial Axis},
    journal = {preprint},
    year = {2026}
}

@misc{gemss2026,
  author       = {Moradian, Arash and Buchele, Felix and M{\"u}ller, Patric and P{\"o}schel, Thorsten},
  title        = {{GEMSS}: {GE}nerator of {M}ulti-{S}phere {S}hapes},
  howpublished = {\url{https://github.com/aqa-arash/GEMSS}},
  year         = {2026}
}

@article{Buchele.2026,
    title = {multisphere: a Python implementation of the Multi Sphere Shape generator (MSS) for DEM simulations},
    journal = {SoftwareX},
    volume = {35},
    pages = {102836},
    year = {2026},
    issn = {2352-7110},
    doi = {10.1016/j.softx.2026.102836},
    author = {Buchele, Felix and Müller, Patric and Moradian, Arash and Pöschel, Thorsten},
}

\end{document}